\documentclass[journal]{IEEEtran}  % Comment this line out if you need a4paper
\usepackage{graphicx}
\usepackage[utf8]{inputenc}
\usepackage[T1]{fontenc}
\usepackage{mathtools}
\usepackage{inputenc}
\usepackage[english]{babel}

\usepackage{amsmath}
\usepackage{amssymb}
\usepackage{amsthm}
\usepackage{comment}
\usepackage{amssymb}
\usepackage{bm}
\usepackage{psfrag}
\usepackage{esdiff}
\usepackage{float}
\usepackage{caption}
\usepackage{subcaption}
\usepackage{bm}
\usepackage{algorithmicx}
\usepackage{algpseudocode}
\usepackage{lipsum}
\usepackage[linesnumbered,ruled,vlined]{algorithm2e}
\usepackage{array}
\usepackage{blindtext}
\usepackage{epstopdf}
\usepackage{ragged2e}
\usepackage{mathrsfs}
\usepackage{tabu}
\graphicspath{ {figures/} }

\usepackage{nomencl}
\usepackage{cite}
\usepackage{lmodern}
\usepackage{enumerate}
\usepackage{tikz}
\usepackage[switch]{lineno}
\usepackage{hyperref}

\IEEEoverridecommandlockouts                              % This command is only needed if 
\DeclareMathOperator*{\argmin}{arg\,min}

\newcommand{\be}{\begin{equation}}
\newcommand{\ee}{\end{equation}}
\newcommand{\bewn}{\begin{equation*}}
\newcommand{\eewn}{\end{equation*}}
\newcommand{\bbmat}{\begin{bmatrix}} 
	\newcommand{\ebmat}{\end{bmatrix}}
\newcommand{\bd}{\begin{displaymath}}
\newcommand{\ed}{\end{displaymath}}
\newcommand{\bea}{\begin{eqnarray}}
\newcommand{\eea}{\end{eqnarray}}
\newcommand{\ba}{\begin{array}}
	\newcommand{\ea}{\end{array}}
\newcommand{\baa}{\begin{array}{ll}}
	\newcommand{\eaa}{\end{array}}
\newcommand{\ds}{\displaystyle}
\newcommand{\bc}{\begin{center}}
	\newcommand{\ec}{\end{center}}
\newcommand{\ben}{\begin{enumerate}}
	\newcommand{\een}{\end{enumerate}}
\newcommand{\bi}{\begin{itemize}}
	\newcommand{\ei}{\end{itemize}}
\newcommand{\bt}{\begin{tabular}}
	\newcommand{\et}{\end{tabular}}
\newcommand{\bte}{\begin{table}}
	\newcommand{\ete}{\end{table}}
\newcommand{\bal}{\begin{align}}
\newcommand{\eal}{\end{align}}

\newcommand{\norm}[1]{\left\lVert#1\right\rVert}   %already in physics package

\makeatletter
\renewcommand\paragraph{\@startsection{paragraph}{4}{\z@}%
	{-2.5ex\@plus -1ex \@minus -.25ex}%
	{1.25ex \@plus .25ex}%
	{\normalfont\normalsize\bfseries}}
\makeatother



\newtheorem{remark}{Remark}
\newtheorem{proposition}{Proposition}
\newtheorem{theorem}{Theorem}
\newtheorem{corollary}{Corollary}
\newtheorem{lemma}[theorem]{\textbf{Lemma}}

\newtheorem{problem}{\textbf{Problem}}

\newcommand{\bR}{\mathbb{R}}

\newcommand{\calC}{\mathcal{C}}

\newcommand{\calE}{\mathcal{E}}
\newcommand{\calF}{\mathcal{F}}
\newcommand{\calG}{\mathcal{G}}

\newcommand{\calK}{\mathcal{K}}

\newcommand{\calO}{\mathcal{O}}

\newcommand{\calV}{\mathcal{V}}
\newcommand{\calW}{\mathcal{W}}

\newcommand{\calY}{\mathcal{Y}}

\newcommand{\pdydx}[2]{\frac{\partial{#1}}{\partial{#2}}}

\begin{document}
	% \linenumbers  (Add line numbers for review purpose)
	\title{\LARGE \bf
		Approximate Time-Optimal Trajectories for Damped Double Integrator in 2D Obstacle Environments under Bounded Inputs%via Bayesian Approach and POMDPs %(Hierarchical)
	}

	\author{Vishnu S. Chipade and Dimitra Panagou% <-this % stops a space
		%\thanks{*This work was not supported by any organization}% <-this % stops a space
		%\thanks{$^{1}$Albert Author is with Faculty of Electrical Engineering, Mathematics and Computer Science,
		%        University of Twente, 7500 AE Enschede, The Netherlands
		%        {\tt\small albert.author@papercept.net}}%
		\thanks{The authors are with the Department of Aerospace Engineering,
			University of Michigan, Ann Arbor, MI, USA;
			{\tt\small (vishnuc,
				dpanagou)@umich.edu}}
		\thanks{This work has been funded by the Center for Unmanned Aircraft Systems (C-UAS), a National Science Foundation Industry/University Cooperative Research Center (I/UCRC) under NSF Award No. 1738714 along with significant contributions from C-UAS industry members.}
	}	
	
	\maketitle
	\thispagestyle{empty}
	\pagestyle{empty}

	%%%%%%%%%%%%%%%%%%%%%%%%%%%%%%%%%%%%%%%%%%%%%%%%%%%%%%%%%%%%%%%%%%%%%%%%%%%%%%%%
	\begin{abstract}
		
		This article provides extensions to existing path-velocity decomposition based time optimal trajectory planning algorithm \cite{kant1986toward} to scenarios in which agents move in 2D obstacle environment under double integrator dynamics with drag term (damped double integrator). Particularly, we extend the idea of tangent graph \cite{liu1992path} to $\calC^1$-Tangent graph to find continuously differentiable ($\calC^1$) shortest path between any two points. $\calC^1$-Tangent graph has continuously differentiable ($\calC^1$) path between any two nodes. We also provide analytical expressions for a near time optimal velocity profile for an agent moving on these shortest paths under damped double integrator with bounded acceleration. 
		
		\textit{Index Terms}---time-optimal control, trajectory planning, double integrator with drag.		
		%Control laws for forming this StringNet and guiding it to a safe area are developed and the performance is analyzed formally. Flocking motion is considered for the adversarial swarm in the presence of rectangular obstacles for which modifications to the existing flocking control laws are provided.  
		
	\end{abstract}
	%and is such that there is no singular point other than the goal location
	
	%%%%%%%%%%%%%%%%%%%%%%%%%%%%%%%%%%%%%%%%%%%%%%%%%%%%%%%%%%%%%%%%%%%%%%%%%%%%%%%%
	\section{Introduction}
	%Accomplishments by a swarm of robots are only limited by our imagination. 
	Trajectory planning is a very important problem for autonomous robots. A significant body of literature is available that solves the problem of trajectory planning \cite{kroger2010literature}.
	For obstacle environments, authors in \cite{kant1986toward} have proposed a method called path-velocity decomposition for trajectory planning. In this article, we consider the idea of path-velocity decomposition \cite{kant1986toward} to find time-optimal trajectory for an agent operating under double integrator dynamics with drag term (damped double integrator) to move from one point to another in an obstacle environment. The path-velocity decomposition approach \cite{kant1986toward} as the name suggests decomposes the trajectory planning problem in two sub-problems: 1) finding a path that avoids collision with static obstacles (path planning), 2) finding velocity profile on the path obtained in 1) to avoid moving obstacles (velocity planning). In \cite{kant1986toward}, V-Graph (vertex graph) is used to find shortest path between to points in polygonal obstacle environment while a simple single integrator model with bounded speed is used to determine the a feasible time optimal velocity profile on the shortest path. 
	
    A better graph representation of an obstacle environment called Tangent graph, which is graph of consisting of common tangents of the polygonal obstacles in an environment is used in \cite{liu1992path} for path planning. However, the paths obtained using Tangent graph are only continuous and not necessarily continuously differentiable ($\calC^1$). We build upon the idea of tangent graph to find $\calC^1$-Tangent graph by considering $\calC^1$ boundaries around the polygonal obstacles in the environment. The $\calC^1$-Tangent graph consists of common tangents to the continuous closed convex boundaries around the polygonal obstacles and there exists a $\calC^1$ path between any two nodes on the $\calC^1$-Tangent graph. This allows us to plan paths for agents moving under second order dynamics. We provide a novel quadratic function to systematically find the common tangents of the $\calC^1$ boundaries of the obstacle. 
    
    Authors in \cite{kim2004shortest} propose two filters: ellipse and convex-hull filter to reduce the search space for finding the shortest paths in the presence of circular obstacles. In this article, we extend the ellipse filter to more generic polygonal obstacles to reduce the search space so that more general obstacle environments can be considered.
    
    In \cite{kunz2012time}, the authors compute a time-optimal velocity profile for following a given path by meticulously keeping track of the lowest bound on the maximum speed and acceleration along the path, while maintaining the overall constant bound on the components of the acceleration vector and velocity vector. However, this approach requires several forward and backward numerical integration of the system dynamics, therefore can be computationally time intensive and also it would result in sub-optimal velocity profiles when euclidean norm constrains are to be considered on the acceleration. Compared to \cite{kunz2012time}, in this article, we consider euclidean norm constraints on the acceleration input which is typical for many under actuated systems and we design near time-optimal velocity\footnote{Near time-optimal in this paper means actual travel time $\tau$ satisfies $\tau^*\le \tau \le (1+\varepsilon )\tau^*$ where $\tau^*$ is the optimal travel time and $\varepsilon<<1$ is a small, positive constant.} profiles for agents moving under damped double integrator for which analytical expressions can be provided saving on computational time.
    	
	In summary, the contributions of this article are as follows:

	\bi
	\item[1)] derivation of continuous ($\mathcal C^0$), near time-optimal velocity profiles for agents moving under damped double integrator dynamics, so that they travel along the shortest paths between the given initial and final points while satisfying acceleration bounds;
	
	\item[2)] derivation of a novel quadratic function to find the common tangents of two continuously differentiable ($\mathcal C^1$), closed, convex curves that approximate the polygonal obstacles;
	
	\item[3)] an extension of ellipse filter \cite{kim2004shortest} for general convex polygonal obstacles to reduce the search space to find the shortest paths between any two nodes on $\calC^1-$Tangent Graph.
	
	%\item[4)] A MILP formulation to modify the assigned trajectories to avoid the inter-agent collisions among the defenders
	\ei

	%\item We provide robust strategy for defenders against the splitting of the attackers in order to guarantee herding of the attackers even after the attackers split into smaller teams or to guarantee capture of the attackers in case none of the attackers stays together.
	%\item For a given initial conditions, we provide set of initial conditions for the attackers from which if they start they are guaranteed to be herded to one of the safe areas.
	
	%\item The extension of standard flocking control laws \cite{olfati2006flocking,dai2014flocking} by designing $\beta$-agents along a $\calC^1$ boundary around the convex polygonal obstacles.
	%; this way, $\calC^1$ velocity profiles for the $\beta$-agents to avoid the obstacles during flocking are obtained.
	
	%superelliptic smooth contour around rectangular obstacles for obstacle avoidance in flocking. This allows smooth velocity profile for the $\beta$-agents while being less conservative around the rectangular obstacles.
	%\item More realistic system dynamics: double integrator with linear drag term.
	
	%\item Novel vector fields around rectangular obstacles to guide the agents to the safe areas
	% \item Optimal assignment of defenders against the teams of attackers
	% \item Compared to our previous work we also assign the desired locations to the defenders optimally such that the total nominal distance is minimized.
	% \item Herding multiple swarms of attackers to safe areas
	% \item Finite time boundedness around the desired positions in the formation 
	% \item Collision avoidance with static and moving obstacles

	\subsection{Organization}
	The rest of the article is structured as follows: Section \ref{sec:math_model} provides the mathematical modeling and problem statement.  In Section \ref{sec:near_time_opitmal_traj}, we give an overview of the path-velocity decomposition method to find near time-optimal trajectory in obstacle environment. Section \ref{sec:c1_tangent_graph} provides details on $\calC^1$-Tangent graph and shortest paths. In Section \ref{sec:near_time_opitmal_traj} we provide analytical expression to find near time optimal velocity profile on given shortest path and the article in concluded in Section \ref{sec:conclusions}.
	
	\section{Modeling and Problem Statement}\label{sec:math_model}
	We consider an agent moving under double integrator dynamics with a quadratic drag term (damped double integrator): 
	\be \label{eq:attackDyn1}
	\baa
	\dot{\mathbf{r}}
	=\mathbf{v},\\
	\dot{\mathbf{v}}
	=\mathbf{u}-C_D\norm{\mathbf{v}}\mathbf{v};
	\eaa
	\ee
$C_D>0$ is the known, constant drag coefficient, $\mathbf{r}=[x\; y]^T$ is the position vector $\mathbf{v}=[v_{x}\; v_{y}]^T$, is the velocity vector, and $\mathbf{u}=[u_{x}\; u_{y}]^T$ is the acceleration of the agent, which serve also as the control input, all resolved in a global inertial frame $\calF_{gi} (\hat {\mathbf{i}}, \hat {\mathbf{j}})$ (see Fig.\ref{fig:shortest_path}). The acceleration $\mathbf{u}$ is bounded as:
 \be \label{eq:input_constraints}
 \norm{\mathbf{u}}\le \bar{u}.
 \ee
	%The acceleration norms are bounded by $\bar{u}_a$ and $\bar{u}$, i.e., $\norm{\mathbf{u}_{ai}}<\bar{u}_a$ and $\norm{\mathbf{u}_{dj}}<\bar{u}$.
	The dynamics in \eqref{eq:attackDyn1} take into account the air drag experienced by the agents modeled as a quadratic function of the velocity. %, which is a realistic modeling of the drag for aerial robots.
	Note also that the above damped double integrator inherently poses a speed bound on the agent under a limited acceleration control, i.e., $\norm{\mathbf{v}}<\bar{v}=\sqrt{\frac{\bar{u}}{C_D}}$, and does not require to consider an explicit constraint on the velocity of the agents while designing bounded controllers as has been done in the literature.
	%The following assumption and the remark are to be used for UAVs in 3D
	
	%	\begin{assumption} \cite{zhu2017distributed} The rotational (attitude for 3D) control of the agents is designed and implemented perfectly,
	%		such that the dynamics can be given as modified double integrator dynamics given in \eqref{eq:attackDyn1}, \eqref{eq:defendDyn1} and the transnational acceleration is the input to be designed. 
	%		%(The assumption and the remark are to be used for UAVs in 3D)
	%	\end{assumption}
	\begin{comment}
	\begin{remark}
	Experimentally, the time-response of the low level inner-loop control can be tuned much faster than that of
	the high-level outer-loop, such that the inner-loop (attitude) dynamics can be neglected in high-level design. This assumption
	is widely applied in guidance design [5], [6].
	\end{remark} 
	\end{comment}
	
	%Generally, buildings have rectangular shape and are the most common obstacles for low altitude flights. For simplicity, these obstacles are usually modeled as circles in literature \cite{panagou2014motion}. 
	We consider $N_o$ static, convex polygonal obstacles $\calO_k$, $ k \in I_o = \{1,2,...,N_o\}$, (grey colored polygons in Fig.~\ref{fig:shortest_path}),
	%	 with their edges aligned with the axes of $\calF_g$,   %and use super-elliptic contours which closely resemble the rectangular shape in order to define a repulsive vector field around them, which are are 
	%	defined as:
	%	\be
	%	\calO_k= \{\mathbf r \in \mathbb{R}^2 | \mathbf{Z}_{ok}\left(\mathbf{r}-\mathbf{r}_{ok}\right)  \le \mathbf{c}_{ok}  \} ,
	%	\ee
	%	or equivalently $\calO_k$ 
	described as the convex hull of their vertices,
	\be
	\calO_k=  Conv\left(\{\mathbf r_{ok}^1, \mathbf r_{ok}^2,...,\mathbf r_{ok}^{M_k}\}\right ),
	\ee
	% \be
	% \calO_k= \{\mathbf r \in \mathbb{R}^2 | \abs{x-x_{ok}} \le \frac{w_{ok}}{2}, \abs{y-y_{ok}} \le \frac{h_{ok}}{2}\} ,
	% \ee
	where $Conv(Q)$ is the convex hull of the points given in the set $Q$, $\mathbf r_{ok}^\ell=[x_{ok}^\ell \; y_{ok}^\ell]^T$ are the positions of the vertices for all $\ell \in \{1,2,...,M_k\}$, $M_k$ is the total number of vertices of $\calO_k$, $k \in I_o$. The boundary of $\calO_k$ is denoted by $\partial \calO_k$. Inspired from \cite{hegde2016multi} and \cite{esquivel2002nonholonomic}, the boundaries $\partial \calO_k$ are inflated by a size of $\rho_{\bar{o}}$ $(> \rho_d)$ to account for safety and agent size. The inflated obstacles are denoted by $\bar{\calO}_k$, and are given as (Fig.~\ref{fig:shortest_path}): $
	\bar{\calO}_k= \calO_k \bigoplus B(\rho_{\bar{o}}),	
	$ where $\bigoplus$ denotes the Minkowski sum of the sets and $B(\rho_{\bar{o}})$ denotes a ball of radius $\rho_{\bar{o}}$ centered at the origin.
	% 	\be
	% 	\baa
	% 	\bar{\calO}_k= \left \{  \mathbf r \in \mathbb{R}^2 | \mathbf{r} \in \bar{\calO}_k^{0} \cup  \sum_{m=1}^{M_k} \bar{\calO}_k^m \right \},	
	% 	\eaa
	% 	\ee
	% 	where $
	% 	\bar{\calO}_k^{0} = \left \{  \mathbf r \in \mathbb{R}^2 |  \mathbf r =  \sum_{l=1}^{2M_k} \bar{\alpha}_{l} \mathbf{r}_{\bar{o}k}^{l},  \forall l:
	% 	\bar{\alpha}_{l} \ge 0, \right. \; $ $\left. \sum_{l=1}^{2M_k} \bar{\alpha}_{l}=1  \right \} ;
	% 	$ $
	% 	\bar{\calO}_k^{m} = \left \{  \mathbf r \in \mathbb{R}^2 |  \norm{\mathbf{r}-\mathbf{r}_{ok}^m} \le \rho_{\bar{o}} \right \} , \forall m \in {1,2,\dots, M_k}.
	% 	$
	\textcolor{black}{The boundary} $\partial \bar{\calO}_k$ of the inflated obstacle $\bar{\calO}_k$ is a $\calC^1$ curve for all $\rho_{\bar{o}}>0$.	%(Fig.~\ref{fig:shortest_path}). % as it consists of straight line segments parallel to the edges of $\calO_k$ and of circular arc segments centered at the vertices of $\calO_k$, which are tangent to the straight line segments.
% 	\begin{figure}
% 		\centering
% 		\includegraphics[width=.9\linewidth,trim={11cm 2.2cm 9.5cm 1.5cm},clip]{figures/problem_formulation.eps}
% 		\caption{\textcolor{black}{Problem Formulation}}
% 		\label{fig:shortest_path}
% 	\end{figure}
	
	We consider the following problem of finding a near time-optimal trajectory for the agent operating under the dynamics in \eqref{eq:attackDyn1} and \eqref{eq:input_constraints}.
	\begin{problem}[Near Time-optimal trajectory]
		Design a control action $\mathbf{u}$ in analytical form such that the agent operating under dynamics \eqref{eq:attackDyn1} and \eqref{eq:input_constraints} travels from an initial position $\mathbf{r}_0$ to final position $\mathbf{r}_f$ in minimum time possible. 
	\end{problem}
	In the next section, we give overview of the near time-optimal trajectory generation algorithm.

	\section{Near Time-optimal Trajectory}\label{sec:near_time_opitmal_traj}

	The near time-optimal trajectory between two given points is obtained by path-velocity decomposition \cite{kant1986toward}, which consists of finding i) the shortest path, and ii) a near time-optimal velocity profile along the shortest path. 
	
	%	Before discussing the MIQP formulation, we discuss shortest path between two points, time optimal velocity profiles on this shortest path and collision checking in next few subsections.
	% The steps involved in this approach are: 1) first design shortest paths for the defenders to these desired goal locations, 2) Pair the defenders with the desired locations based on the minimum time required to travel the shortest path corresponding to the assigned goal location, 
	% %using a polynomial time Hungarian Algorithm \cite{munkres1957algorithms}, 
	% 3) Design velocity profiles for the defenders to trace the path corresponding to the assigned goal location while avoiding the potential collisions with other defenders. 

	i) \textit{Shortest paths}:
	We propose an approach wherein the kinematic and safety constraints are directly incorporated while constructing a special representation of the obstacle environment called $\calC^1$-Tangent graph, inspired from the idea of tangent graph \cite{liu1992path}.
	%, and differently to the visibility graph approaches for finding curvature-constrained shortest paths \cite{maini2016path}. 
	The $\calC^1$-Tangent Graph ($\calG_{ct}$) consists of $\mathcal C^1$ paths between any two nodes. For an obstacle environment with convex polygonal obstacles $\calO_k$, the construction of the $\calC^1$-Tangent graph involves connecting the enlarged obstacles $\bar{\calO}_k$ with their common tangents; then, the points at which these tangents touch the boundaries $\partial \bar{\calO}_k$ serve as the nodes on $\calG_{ct}$. Using $\calG_{ct}$, one can find a path between any two nodes on $\calG_{ct}$ in the obstacle-free environment $\textcolor{black}{\calW_{free}}=\calW \backslash (\cup_{k=1}^{N_o} \calO_k) $ using Dijkstra's algorithm \cite{cormen2009introduction}. The details of $\calC^1$-Tangent graph and shortest path are provided in section \ref{sec:c1_tangent_graph}
	%The graph $\calG_{ct}$ can be considered as a roadmap in which there exist a path in obstacle free environment $\textcolor{black}{\calW_{free}}$ between any two nodes on the roadmap.

% 	The details about the $\calC^1$-tangent graph are provided in Appendix \ref{append:shortest_path}. Specifically, we provide a novel quadratic function whose zeros give the locations at which the lines connecting the obstacles are also tangent to the obstacles. We also extend the ellipse filter  \cite{kim2004shortest} to the case of general convex polygonal obstacles; the filter is used to reduce the search space in $\calG_{ct}$ while finding the shortest paths (see Appendix \ref{append:shortest_path} for details).
	
 ii) \textit{Near time-optimal velocity profiles}:
	%	To travel on the shortest paths, we seek time-optimal velocity profiles with bounded acceleration. 
	%developed a MILP formulation for coordinating velocities of multiple agents on given paths, while ensuring minimum time to reach the goals and avoiding collisions. The authors
	% With this bound the  Although this assumption can be applicable to dubin like vehicles on the dubins paths, it cannot model multi-rotor vehicles whose dynamics can be closely modeled by double integrator in the original state space. 
	%Furthermore, computational time to solve the optimization problem would be a bottleneck while applying this approach to problems which require solving the problem more frequently or with time critical application such as ours. 
	%We provide a two tier approach to solve the velocity coordination problem.
	%% READ UP TO HERE
	The shortest path between two points obtained using $\calC^1$-Tangent graph consists of straight line segments and circular arcs.
	%as shown in Fig.~\ref{fig:shortest_path}. 
	Our goal is to obtain analytical expression for a near time-optimal velocity profile for each defender on its desired shortest path under bounded acceleration, to avoid huge computational cost that numerical integration schemes incur. We build on the approach described in \cite{peng2005coordinating,kunz2012time}, adopting however the following different assumptions: 1) we consider a quadratic drag term in the dynamics, which apart from being a realistic assumption for aerial vehicles also imposes an inherent bound on the velocity; 2) we consider a constant bound on the norm of the acceleration instead of bounds on each component, as this is a more realistic assumption for under-actuated, multi-rotor vehicles such as quadrotors or hexacopters. We provide analytical expressions for the velocity bounds along a path, and for the time to traverse the path. The technical details on finding the near time-optimal velocity profiles on the shortest paths are given in Section \ref{sec:time-optimal_velocity_profile1}.\\
% 	\begin{remark}
% 	 %The authors in \cite{peng2005coordinating} used double-integrator dynamics that governs the distance $\gamma$ along the path and the speed $v$ along the path with a constant bound on $\dot{v}$. However, for paths with varying curvature, a constant bound on $\dot{v}$ would be a conservative assumption. 
% 	\end{remark}
\section{$\calC^1$-Tangent Graph and Shortest Path}\label{sec:c1_tangent_graph}
		The $\calC^1$-Tangent Graph $\calG_{ct}$ requires finding common tangents of the pairwise boundary curves $\partial \bar{\calO}_{k}$, for all $k \in I_o$. \textcolor{black}{We know that
% 		\begin{lemma}
			any two closed convex curves $\calC_1$, $\calC_2$ in 2D with non-overlapping interiors will have a maximum of four common tangent lines (Corollary 3.2 in \cite{czedli2016note}).
% 		\end{lemma}
% 		\begin{proof}
% 			The proof follows from Corollary 3.2 in \cite{czedli2016note}.
% 			% Let us consider two convex curves $\calC_1$ and $\calC_2$
% 			% Convex curve is a curve in Euclidean space which lies entirely on one side of its all tangent lines. For any common tangent to   
% 		\end{proof}
		}
		To find the common tangents of the approximated convex obstacles $\bar{\calO}_{k}$, we formulate a novel quadratic function whose zeros give the locations at which the common tangents are tangent to the boundary curves $\partial \bar{\calO}_{k}$ and $\partial \bar{\calO}_{k'}$. The function is developed for general convex curves as detailed in Lemma \ref{lem:common_tangent_fun}. 
		
		\begin{lemma}[Common Tangents] \label{lem:common_tangent_fun}
			Consider any two closed convex curves $\calC_1$, $\calC_2$ parameterized by $\calY_1:[0,\Gamma_1] \rightarrow \bR^2$ and $\calY_2:[0,\Gamma_2] \rightarrow \bR^2$, where $\calY_{\imath}(\gamma_{\imath})=[x_{\imath}(\gamma_{\imath}), \; y_{\imath}(\gamma_{\imath})]^T$ for ${\imath}=\{1,2\}$. Define 
			$
			f(\gamma_1,\gamma_2)=\left (m(\gamma_1,\gamma_2)-m_1(\gamma_1) \right)^2+\left (m(\gamma_1,\gamma_2)-m_2(\gamma_2) \right)^2,
			$
			where $m(\gamma_1,\gamma_2)=\frac{y_2(\gamma_2)-y_1(\gamma_1)}{x_2(\gamma_2)-x_1(\gamma_1)}$ is the slope of the line joining $\calY_1(\gamma_1)$ and $\calY_2(\gamma_2)$, and for ${\imath}=\{1,2\}$, $m_{\imath}(\gamma_{\imath})=\frac{y_{\imath}^{'}(\gamma_{\imath})}{x_{\imath}^{'}(\gamma_{\imath})}$ is the slope of the tangent to $\calC_{\imath}$ at $\calY_{\imath}(\gamma_{\imath})$, where $'$ denotes the derivative with respect to $\gamma_{\imath}$ (see Fig. \ref{fig:common_tangent}). The function $f$ is locally convex and $f \ge 0, \forall [\gamma_1,\;\gamma_2]^T \in [0,\;\Gamma_1] \times [0,\;\Gamma_2]$. The solutions $\bm{\gamma}^{*}=[\gamma_1^{*},\;\gamma_2^{*}]^T$ to $f(\gamma_1,\gamma_2)=0$ give the points through which common tangents to $\calC_1$ and $\calC_2$ pass.
		\end{lemma}
		\begin{figure}[h]
			\centering
			\includegraphics[width=.7\linewidth,trim={5.5cm 2.2cm 5.5cm 1.7cm},clip]{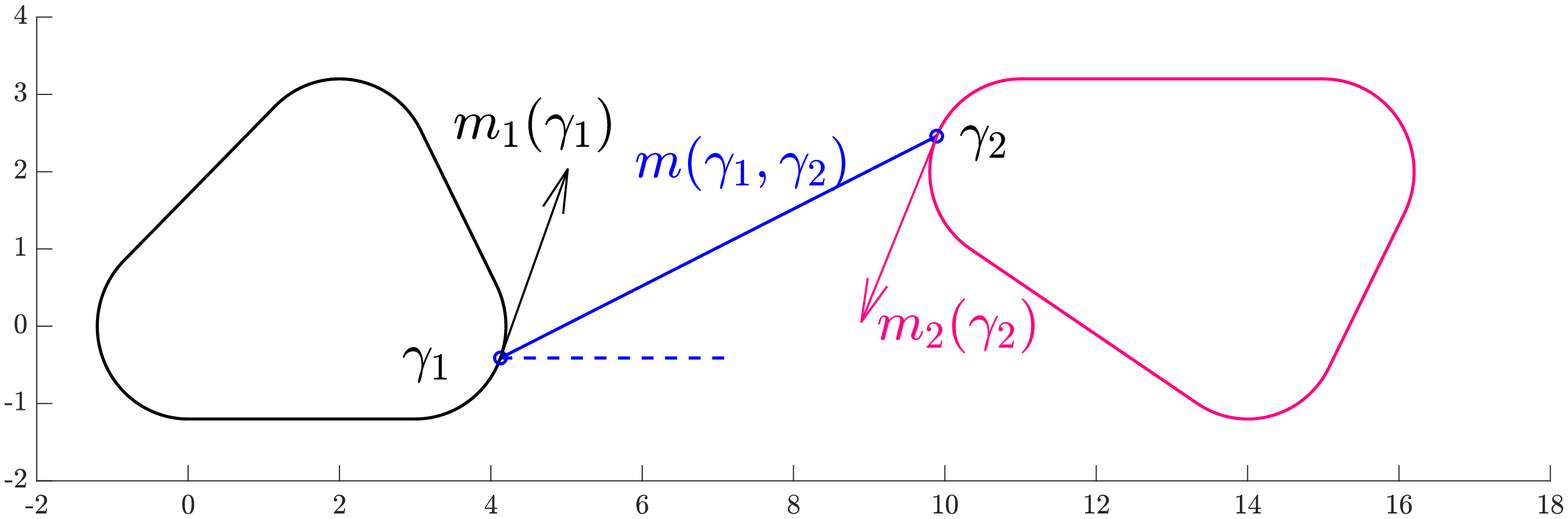}
			\caption{Common Tangent}
			\label{fig:common_tangent}
		\end{figure}
		\begin{proof}
		We first prove the local convexity property of $f(\gamma_1,\gamma_2)$. The gradient of $f(\gamma_1,\gamma_2)$ is:
		\be
		\nabla f(\gamma_1,\gamma_2) =\bbmat 2\bigl((m-m_1)(m_{\gamma_1}-m_1^{\prime})+(m-m_2)(m_{\gamma_1})\bigr) \\ 2\bigl((m-m_1)(m_{\gamma_2})+(m-m_2)(m_{\gamma_2}-m_2^{\prime})\bigr) \ebmat
		\ee
		where $m_{\gamma_{\imath}}=\pdydx{m}{\gamma_{\imath}}$, for $\imath=\{1,2\}$.
		We notice that the points $(\gamma_1^{*},\;\gamma_2^{*})$ for which $f(\gamma_1^{*},\;\gamma_2^{*})=0$ are also the points where the gradient of $f$ vanishes because $m(\gamma_1^{*},\;\gamma_2^{*})=m_1(\gamma_1^{*})=m_2(\gamma_2^{*})$. The jacobian of $f$ is:
			\be
		\nabla^2 f(\gamma_1,\gamma_2) =\bbmat f_{\gamma_1^2 } & f_{\gamma_1 \gamma_2} \\f_{\gamma_1 \gamma_2} & f_{\gamma_2^2}  \ebmat
		\ee
		where 
		\be
		\baa
		f_{\gamma_1^2}=&2\bigl ( (m_{\gamma_1}-m_1^{'})^2 + (m-m_1)(m_{\gamma_1^2}-m_1^{''}) +(m_{\gamma_1})^2\\
		&(m-m_2)m_{\gamma_1^2}\bigr )\\
		f_{\gamma_2^2}=&2\bigl ( (m_{\gamma_2}-m_2^{'})^2 + (m-m_2)(m_{\gamma_2^2}-m_1^{''}) +(m_{\gamma_2})^2\\
		&(m-m_1)m_{\gamma_2^2}\bigr )\\
		f_{\gamma_1 \gamma_2}=&2\bigl ( (m_{\gamma_1}-m_1^{'})m_{\gamma_2} + (m_{\gamma_2}-m_2^{'})m_{\gamma_1}\\
		&(2m-m_1-m_2)m_{\gamma_1 \gamma_2}\bigr )
		\eaa
		\ee
		We have the determinant:
		\be
		det (\nabla^2 f(\gamma_1,\gamma_2))|_{(\gamma_1^{*},\;\gamma_2^{*})}=\Bigl (4 (m_{\gamma_1})^2(m_{\gamma_2})^2\Bigr)|_{(\gamma_1^{*},\;\gamma_2^{*})} \ge 0.
		\ee
		This implies $f(\gamma_1,\gamma_2)$ is locally convex around the points $(\gamma_1^{*},\;\gamma_2^{*})$.
		It is very easy to see from the definition of $f(\gamma_1^{*},\;\gamma_2^{*})$ why the solutions to $f(\gamma_1,\;\gamma_2)=0$ give coordinates of the points at which the two boundaries have a common tangent. 
		\end{proof}

		A nonlinear solver can provide the solutions to $f(\gamma_1,\gamma_2)=0$ with appropriate initial conditions. 
		
		Let $P_r(\bar{\calO}_k)$ denote the perimeter of the boundary $\partial \bar{O}_k$. To find the common tangents, we parameterize the boundary $\partial \bar{O}_k$ by $\calY_k: [0,\;P_r(\bar{\calO}_k)] \rightarrow \bR^2$ given as:
		\be
		\baa
		\calY_k(\gamma) = &\displaystyle \sum_{\ell=1}^{M_k} \tilde{\mathbb{B}}_{k}^\ell(\gamma)\left (\mathbf{r}_{ok}^\ell+\rho_{\bar{o}} \hat{\mathbf{o}}(\psi_k^\ell(\gamma))\right ) \\
		&+ \bar{\mathbb{B}}_{k}^\ell(\gamma) \left( \mathbf{r}_{\bar{o}k}^{2\ell}+\alpha_{k}^\ell(\gamma) (\mathbf{r}_{\bar{o}k}^{2\ell+1}-\mathbf{r}_{\bar{o}k}^{2\ell})\right) 
		\eaa
		\ee
		where $\alpha_{k}^\ell(\gamma)=\frac{\gamma-\gamma_{\bar{o}k}^{2\ell}}{\gamma_{\bar{o}k}^{2\ell+1}-\gamma_{\bar{o}k}^{2\ell}}$, $\tilde{\mathbb{B}}_{k}^\ell(\gamma)$ is 1 when $\gamma_{\bar{o}k}^{2 \ell-1}<\gamma<\gamma_{\bar{o}k}^{2\ell}$ and 0 otherwise, $\bar{\mathbb{B}}_{k}^\ell(\gamma)$ is 1 when $\gamma_{\bar{o}k}^{2\ell}<\gamma<\gamma_{\bar{o}k}^{2\ell+1}$ and 0 otherwise, and  $\psi_k^\ell(\gamma)=\psi_{\bar{o}k}^{2\ell-1}+\alpha_k^\ell(\gamma)(\psi_{\bar{o}k}^{2\ell}-\psi_{\bar{o}k}^{2\ell-1})$,  where $\gamma_{\bar{o}k}^l$ and $\psi_{\bar{o}k}^l$ are the parameters as defined in Fig.~\ref{fig:shortest_path} corresponding to the common points $\mathbf{r}_{\bar{o}k}^l$ of circular and straight line segments on $\partial \bar{\calO}_k$. Since all the common tangents would be the ones on the circular arcs of the approximated obstacles $\bar{\calO}_k$ and $\bar{\calO}_{k'}$, we initialize $\gamma_k$, $\gamma_{k'}$ to the values that correspond to the circular arcs of the boundaries $\partial \bar{\calO}_k$ and $\partial \bar{\calO}_k'$ in Lemma \ref{lem:common_tangent_fun}. 	
		
		To reduce the search space for finding the shortest paths in the presence of circular obstacles, \cite{kim2004shortest} proposes two filters: ellipse and convex-hull filter. We extend the ellipse filter to general convex polygonal obstacles. Let $L_s(\mathbf{r}_0,\mathbf{r}_f)$ be the length of the geodesic between the two points $\mathbf{r}_0$ and $\mathbf{r}_f$ (solid red path in Fig. \ref{fig:shortest_path}), and $L(\mathbf{r}_0,\mathbf{r}_f)$ be length of the straight line between $\mathbf{r}_0$ and $\mathbf{r}_f$ (dotted pink line in Fig. \ref{fig:shortest_path}). 
		\begin{figure}[h]
			\centering
		\includegraphics[width=.8\linewidth,trim={8.6cm 1.6cm 7.5cm 1.6cm},clip]{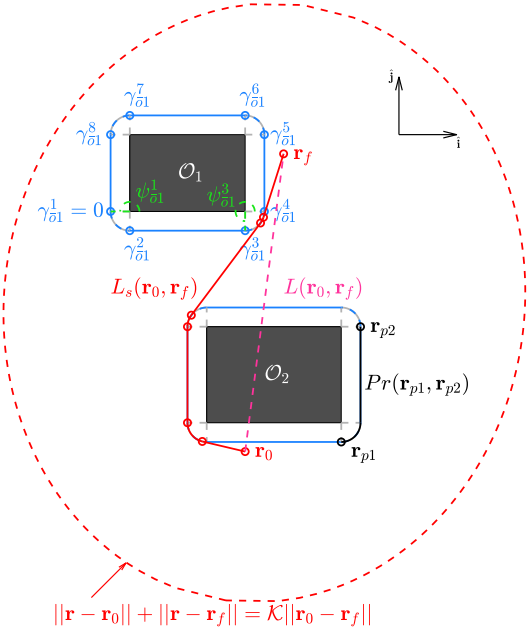}
			\caption{Shortest Path in an Obstacle Environment}
			\label{fig:shortest_path}
		\end{figure}
		Denote the minimum distance between the centroid of $\calO_k$ and any point on $\bar{\calO}_k$ by $\rho_k^{in}$. Consider two points  $\mathbf{r}_{p_1}$  and  $\mathbf{r}_{p_2}$ on $\partial \bar{\calO}_k$. Let $P_r(\mathbf{r}_{p_1},\mathbf{r}_{p_2})$ be the perimeter of boundary curve between  $\mathbf{r}_{p_1}$  and  $\mathbf{r}_{p_2}$, i.e., the smallest distance between the points $\mathbf{r}_{p_1}$ and $\mathbf{r}_{p_2}$ along the boundary $\partial \bar{\calO}_k$ (Fig.~\ref{fig:shortest_path}).
		
		\begin{lemma}One has:
			$L_s(\mathbf{r}_0,\mathbf{r}_f) \le \calK L(\mathbf{r}_0,\mathbf{r}_f)$, where $$\calK = \displaystyle \max_{k \in I_o} \left \{ \max_{\mathbf{r}_{p_1},\mathbf{r}_{p_2} \in \partial \bar{\calO}_k} \left \{\frac{P_r(\mathbf{r}_{p_1},\mathbf{r}_{p_2})}{R_{p_1}^{p_2}}  \right \} \right \} \le \max_{k \in I_o} \left\{\frac{P_r( \bar{\calO}_k)}{4\rho_k^{in}} \right \}$$ with $P_r(\bar{\calO}_k)$ being the perimeter of $\bar{\calO}_k$.
		\end{lemma}
		\begin{proof}
			The proof follows from the proof of Lemma 1 in \cite{kim2004shortest}. The factor $\calK$ is the maximum ratio of the shortest perimeter with the distance of two points on the boundary of any obstacle. 	
		\end{proof}
		
		The ellipse filter to find the shortest path can be applied by considering the obstacles lying completely inside the ellipse as given in Corollary 1 \cite{kim2004shortest}.
		\begin{corollary}
			The shortest path between two points $\mathbf{r}_p$ and $\mathbf{r}_q$ lies inside the ellipse defined by $
			\norm{\mathbf{r}_p-\mathbf{r}}+\norm{\mathbf{r}_q-\mathbf{r}}=\calK_m L_s({\mathbf{r}_p,\mathbf{r}_q}),$
			where $\calK_m = \max_{k \in I_o} \left\{\frac{P_r( \bar{\calO}_k)}{4\rho_k^{in}} \right\} $.
		\end{corollary}
		
		The $\calC^1$-Tangent graph $\calG_{ct}$ can be computed offline, and stored for finding the shortest paths using Dijkstra's algorithm \cite{cormen2009introduction}. 
		Any shortest path $\mathbf{P}_\jmath$, obtained using the $\calC^1$-Tangent graph, is associated with mappings $\mathscr{P}_\jmath:[0,\Gamma^\jmath] \rightarrow \bR^2$ and $\vartheta_\jmath:$ $[0,\Gamma^\jmath]\rightarrow [0,2\pi]$, where $\Gamma^\jmath$ is the total length of the path $\mathbf{P}_\jmath$. Here $\mathscr{P}_\jmath(\gamma_\jmath)$ gives the Cartesian coordinates, and $\vartheta_\jmath(\gamma_\jmath)$ gives the direction of the tangent to the path at the location reached after traveling $\gamma_\jmath$ distance along the path from the initial position. 
		%These mappings are used in the design of the desired open formation, which in turn is based on the attackers' shortest path.\\
		
		When there are more than one agents moving on different shortest paths found using the $\calC^1$-Tangent graph, they might collide with each other. Whether they will collide with each other or not depends on whether the corresponding shortest paths intersect with each other. We have the following result regarding the intersection of the shortest paths.
		\begin{lemma}\label{lem:unique_collision_segment}
	Let $\mathbf{P}_1$ be the shortest path between the points $\mathbf{r}_{11}$ and $\mathbf{r}_{12}$,  and $\mathbf{P}_2$ be the shortest path between the points $\mathbf{r}_{21}$ and $\mathbf{r}_{22}$. If $\mathbf{P}_1$ and $\mathbf{P}_2$ are obtained using the $\calC^1$-Tangent Graph then they intersect at most once.
	\end{lemma}
		\begin{proof} %We show this by considering the intersection segments on $\mathbf{P}_1$ and $\mathbf{P}_2$.
			Case (1) - Both paths $\mathbf{P}_1$ and $\mathbf{P}_2$ are straight-line segments, i.e., there is no circular segment on either of these paths. This case is trivial as two lines intersect at most once, so any two line-segments of these paths would intersect each other at most once. 
			%Hence there would be maximum of one intersection/collision segment on such paths $\mathbf{P}_1$ and $\mathbf{P}_2$.
			
			%READ UP TO HERE
			
			Case (2) - At least one of the two shortest paths has one or more circular segments.
			We prove this case by contradiction. Let us assume that there are 2 distinct collision segments on the two paths $\mathbf{P}_1$ and $\mathbf{P}_2$. This is possible when these two paths converge towards each other, then diverge and then converge again. A possible scenario for this to happen is shown in Fig \ref{fig:shortestPathContradiction}. Suppose $\mathbf{P}_1$ (solid blue, $\Gamma(\mathbf{P}_1)=26.89 m$) and $\mathbf{P}_2$ (solid red, $\Gamma(\mathbf{P}_2)=25.00 m$) are the two shortest paths joining $\mathbf{r}_{11}$ to $\mathbf{r}_{12}$ and $\mathbf{r}_{21}$ to $\mathbf{r}_{22}$, respectively. These two paths as can be seen in Fig. \ref{fig:shortestPathContradiction} have two intersections. 
			\begin{figure}
				\centering
				\includegraphics[width=.9\linewidth,trim={5cm 2.7cm 3.5cm 7.5cm},clip]{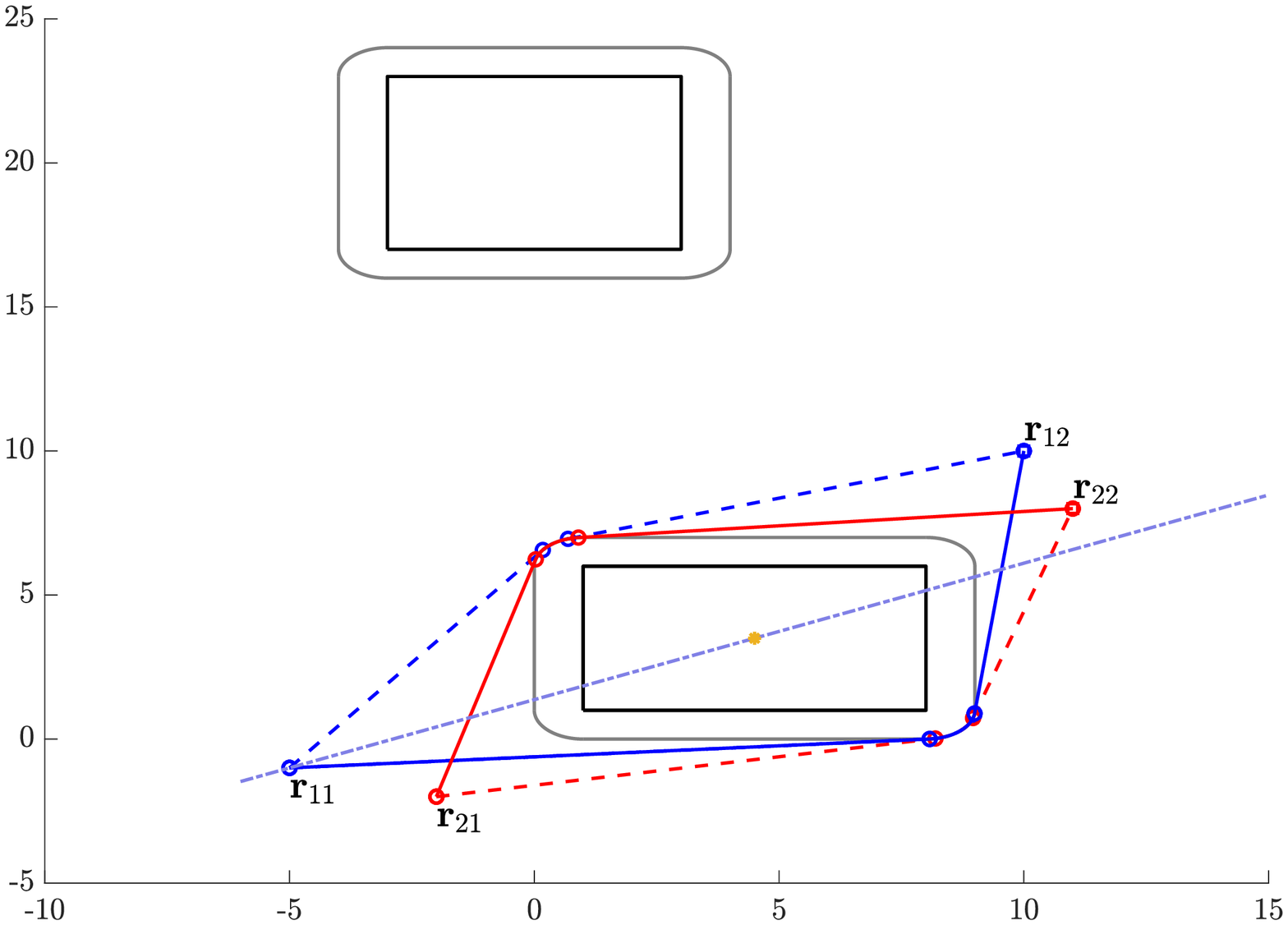}
				\caption{Shortest path intersection}
				\label{fig:shortestPathContradiction}
			\end{figure}
			However, there exists a path $\bar{\mathbf{P}}_1$ (dotted blue) with total length $\Gamma(\bar{\mathbf{P}}_1)=22.61m$ which would be the shortest path between $\mathbf{r}_{11}$ to $\mathbf{r}_{12}$. This path has only one intersection with $\mathbf{P}_2$. Similar argument can be applied to $\mathbf{P}_2$ and $\bar{\mathbf{P}}_2$ (dotted red, $\Gamma(\bar{\mathbf{P}}_2)=24.97 m$). The same observation can be made in case of multiple obstacles. This implies that the paths $\mathbf{P}_1$ and $\mathbf{P}_2$ intersect at most once.
		\end{proof}

		\section{Near time-optimal velocity profiles under bounded acceleration}\label{sec:time-optimal_velocity_profile1}
		Consider the shortest path $\mathbf{P}_0^f$ from $\mathbf{r}_0$ to $\mathbf{r}_f$ consisting of $N_{\mathbf{P}}^c$ circular arc segments, and $N_{\mathbf{P}}^s = N_{\mathbf{P}}^c+1$ straight line segments. Let $v_{2l-1}$ and $v_{2l}$ be the speeds at the endpoints of the $l^{th}$ circular segment, respectively, while moving forward along the path starting at the initial position. We formulate the problem of finding a near time-optimal velocity profile on the path from point $\mathbf{r}_0$ to point $\mathbf{r}_f$ under bounded acceleration as a problem of finding the feasible terminal speed vector $\tilde{\mathbf{v}}=[v_1,\; v_2, ...v_{2N_{\mathbf{P}}^c+1}, v_{2N_{\mathbf{P}}^c+2}]$ that minimizes the total time of travel. In the following, we discuss time-optimal control on straight line and and near time-optimal control on circular segments with specified terminal speeds under acceleration constraint.

		\underline{Minimum time on a straight line segment}:
		For a straight line segment of length $\Gamma_1$ and the terminal speeds $v(0)=v_{0}$ and $v(\Gamma_1)=v_{f}$, the time-optimal control operates at either extremes \cite{bryson1969applied,peng2005coordinating}, i.e., $\norm{\mathbf{u}} = \bar{u}$. The minimum time under this control action after integration is:
		\be
		\baa
		\tau_1 (\Gamma_1,v_0,v_f) =& \frac{1}{\sqrt{\bar{u}C_D}} \left(\tanh^{-1} \left(\frac{v_{sw}}{\bar{v}_d} \right) + \tan^{-1} \left(\frac{v_{sw}}{\bar{v}_d} \right)\right.\\
		& \left. - \tanh^{-1} \left(\frac{v_0}{\bar{v}_d} \right) -\tan^{-1} \left(\frac{ v_f}{\bar{v}_d} \right)\right),
		\eaa
		\ee
		where $v_{sw}=\sqrt{\frac{(\lambda-1)\bar{u}}{(\lambda+1)C_D}}$ is the speed at which the control action switches from one extreme to the other, and $\lambda=\left (\frac{\bar{u}+C_Dv_f^2}{\bar{u}-C_Dv_0^2} \right )e^{2C_D\Gamma_1}$.
		% $\bar{v}_d=\sqrt{\frac{\bar{u}}{C_D}}$. 
		
		\underline{Approximate minimum time on a circular segment}:
		The dynamics of an agent along the path parametrized by $\bm{f}(\gamma)$ is:
		\be
		\arraycolsep=0.5mm
		\baa
		\dot{\mathbf{r}}
		=&\bm{f}'(\gamma) v,\\
		\ddot{\mathbf{r}}=&\bm{f}'(\gamma) \dot{v} +\bm{f}''(\gamma) v^2,
		\eaa
		\ee
		where $\gamma$ is the distance traveled along the path from the initial position, and $v$ is the speed. Similar to \cite{peng2005coordinating}, we consider double integrator dynamics along the path. But instead of constant bound on the acceleration along the path $\dot{v}$, we consider state dependent constraints on $\dot{v}$ to ensure the acceleration $\mathbf{u}$ satisfies the constraints \eqref{eq:input_constraints}, i.e.,
		\be \label{eq:pathDyn}
		\arraycolsep=0.5mm
		\baa
		\dot{\gamma}=&v, \\
		\dot{v}=&a,
		\eaa
		\ee
		where $\norm{\bm{f}'(\gamma) a + \bm{f}''(\gamma) v^2 +C_D \bm{f}'(\gamma)v^2}<\bar{u}$.	
		For a circular path of radius $\rho_{\bar{o}}$ centered at $\mathbf{r}_c$, we have:
		\be
		\arraycolsep=1pt
		\baa
		%\bm{f}(\gamma)=\mathbf{r}_c+\rho_{\bar{o}}\bbmat \cos(\alpha_0+\frac{\gamma}{\rho_{\bar{o}}}) \\ \sin(\alpha_0+\frac{\gamma}{\rho_{\bar{o}}}) \ebmat\\
		\bm{f}'(\gamma)=\bbmat -\sin(\alpha_0+\frac{\gamma}{\rho_{\bar{o}}}) \\ \cos(\alpha_0+\frac{\gamma}{\rho_{\bar{o}}}) \ebmat, \;
		\bm{f}''(\gamma)=-\frac{1}{\rho_{\bar{o}}} \bbmat \cos(\alpha_0+\frac{\gamma}{\rho_{\bar{o}}}) \\ \sin(\alpha_0+\frac{\gamma}{\rho_{\bar{o}}}) \ebmat
		\eaa
		\ee
		where $\alpha_0$ is the orientation of the initial position vector with respect to $\mathbf{r}_c$. 
		The constraint then reads:
		\be \label{eq:pathAccelConstraints1}
		-C_D v^2 - \sqrt{\bar{u}^2-\frac{v^4}{\rho_{\bar{o}}^2}} \le a \le -C_D v^2 + \sqrt{\bar{u}^2-\frac{v^4}{\rho_{\bar{o}}^2}}.
		\ee
		Suppose a defender has to travel a circular segment of length $\Gamma_2$ and the terminal speeds on this path are $v(0)=v_{0}$ and $v(\Gamma_2)=v_{f}$. The time-optimal control for the trajectories of \eqref{eq:pathDyn} 
		%\textbf{UNCLEAR:} 
		is a bang-bang control with acceleration $a$ operating at its extreme values based on a switching condition \cite{bryson1969applied}. However, this requires integration of the system \eqref{eq:pathDyn} under the extreme accelerations given in \eqref{eq:pathAccelConstraints1} to find the switching condition and required total minimum time. This involves inverting hypergeometric functions, which would be computationally intensive, so we approximate the acceleration bounds as:
		\be  \label{eq:pathAccelConstraints2}
		-C_D v^2 - \left (\frac{\bar{u}^2\rho_{\bar{o}}^2-v^4}{\bar{u}\rho_{\bar{o}}^2} \right) \le a \le -C_D v^2 + \left (\frac{\bar{u}^2\rho_{\bar{o}}^2-v^4}{\bar{u}\rho_{\bar{o}}^2} \right).
		\ee
		It is easy to verify that the acceleration $a$ satisfying the constraints in \eqref{eq:pathAccelConstraints2} also satisfies the constraints in \eqref{eq:pathAccelConstraints1}.	
		\begin{comment}
		%Solution of system dynamics \dot{v}=u-\frac{v^4}{u\rho_{\bar{o}}^2}
		\be
		\baa
		s=&s_0+\frac{\rho_{\bar{o}}}{2} \left( \tanh^{-1} \left(\frac{v^2}{(\bar{v}^c)^2} \right)-  \tanh^{-1}\left(\frac{v_0^2}{(\bar{v}^c)^2} \right) \right)\\
		t=&\frac{1}{2}\sqrt{\frac{\rho_{\bar{o}}}{\bar{u}}} \left(\tan^{-1}\left( \frac{v}{\bar{v}^c}\right)+\tanh^{-1}\left( \frac{v}{\bar{v}^c}\right) \right) \\
		& -\frac{1}{2}\sqrt{\frac{\rho_{\bar{o}}}{\bar{u}}} \left(\tan^{-1}\left( \frac{v_0}{\bar{v}^c}\right)+\tanh^{-1}\left( \frac{v_0}{\bar{v}^c}\right) \right) 
		\eaa
		\ee
		\end{comment}
		
		The minimum time required to travel a circular segment of length $\Gamma_2$ with terminal speeds $v_0$ and $v_f$ after integrating \eqref{eq:pathDyn} with approximate extreme accelerations in \eqref{eq:pathAccelConstraints2} is:
		\be \label{eq:circ_path_opt_time}
		\arraycolsep=0.2mm
		\baa
		\tau_2(\Gamma_2,v_0,v_f)=&\frac{1}{\lambda_3}\left ( \frac{\tan^{-1}\left( \lambda_1^{\circ}v_{sw}\right)}{\lambda_1}+\frac{\tanh^{-1}\left(\lambda_2^{\circ}v_{sw} \right)}{\lambda_2}\right) \\
		& -\frac{1}{\lambda_3}\left ( \frac{\tan^{-1}\left(\lambda_1^{\circ}v_{0}\right)}{\lambda_1}+\frac{\tanh^{-1}\left(\lambda_2^{\circ}v_{0} \right)}{\lambda_2}\right)\\
		& + \frac{1}{\lambda_3}\left ( \frac{\tan^{-1}\left(\lambda_2^{\circ}v_{f}\right)}{\lambda_2}+\frac{\tanh^{-1}\left(\lambda_1^{\circ}v_{f} \right)}{\lambda_1}\right) \\
		& - \frac{1}{\lambda_3}\left ( \frac{\tan^{-1}\left(\lambda_2^{\circ}v_{sw}\right)}{\lambda_2}+\frac{\tanh^{-1}\left(\lambda_1^{\circ}v_{sw} \right)}{\lambda_1}\right),
		\eaa
		\ee
		
		where $v_{sw}=\sqrt{\frac{\kappa_1 \left( e_{\kappa} +1 \right) + \sqrt{(\kappa_1(e_{\kappa}+1))^2-(e_{\kappa}-1)^2(\kappa_1^2-\kappa_2^2)}}{2(e_{\kappa}-1)}}$ is the speed at witch the control action switches from one extreme to the other,  $e_{\kappa}=e^{\frac{\kappa}{\kappa_0}}$, $\kappa=\Gamma_2+2\kappa_0\tanh^{-1}\left(\frac{\kappa_2-2v_0^2}{\kappa_1}\right)-2\kappa_0\tanh^{-1}\left(\frac{\kappa_2-2v_f^2}{\kappa_1}\right)$, $\kappa_0=\frac{\rho_{\bar{o}}}{2\lambda_0}$, $\kappa_1=\rho_{\bar{o}} \bar{u}\lambda_0$, $\kappa_2=C_D\rho_{\bar{o}}^2 \bar{u}$, $\lambda_1=\sqrt{\rho_{\bar{o}} \left( \lambda_0-\rho_{\bar{o}} C_D \right)}$, $\lambda_2=\sqrt{\rho_{\bar{o}} \left( \lambda_0+\rho_{\bar{o}} C_D \right)}$, $ \lambda_3=\frac{\lambda_0}{\rho_{\bar{o}}}\sqrt{\frac{\bar{u}}{2}}$,  $\lambda_1^{\circ}=\frac{\sqrt{2}}{\lambda_1 \sqrt{\bar{u}}}$, $\lambda_2^{\circ}=\frac{\sqrt{2}}{\lambda_2 \sqrt{\bar{u}}}$, and $\lambda_0=\sqrt{\rho_{\bar{o}}^2 C_D^2+4}$.
		%	\be \label{eq:circ_path_opt_time}
		%	\baa
		%	\tau_c(v_0,v_f)=&\frac{1}{4}\sqrt{\frac{\rho_{\bar{o}}}{\bar{u}}} \left [ 2 \left ( \log \left(\frac{\bar{v}^c+v_{sw}}{\bar{v}^c-v_{sw}} \right) +2 \tan^{-1}\left(\frac{v_{sw}}{\bar{v}^c}\right) \right) \right.\\
		%	&- \left ( \log \left(\frac{\bar{v}^c+v_{0}}{\bar{v}^c-v_{0}} \right) +2 \tan^{-1}\left(\frac{v_{0}}{\bar{v}^c}\right) \right) \\
		%	& \left. - \left ( \log \left(\frac{\bar{v}^c+v_{f}}{\bar{v}^c-v_{f}} \right) +2 \tan^{-1}\left(\frac{v_{f}}{\bar{v}^c}\right) \right)  \right ]
		%	\eaa
		%	\ee
		%	where  $v_{sw}=\bar{v}^c\sqrt{\tanh \left (\frac{P^c}{\rho_{\bar{o}}} +\frac{1}{4}\log \left ( \frac{((\bar{v}^c)^2+v_{f}^2)}{((\bar{v}^c)^2-v_{f}^2)} \frac{((\bar{v}^c)^2+v_{0}^2)}{((\bar{v}^c)^2-v_{0}^2)} \right) \right)}$. If one of the 
		
		% where $v_{sw}=\bar{v}^c\sqrt{\tanh \left (\frac{S}{\rho_{\bar{o}}} +\frac{1}{2}\tanh^{-1} \left( \frac{v_{0}^2}{\bar{v}_d^2}\right)  +\frac{1}{2}\tanh^{-1} \left( \frac{v_{f}^2}{\bar{v}_d^2}\right ) \right)}$
		\begin{comment}
		\begin{figure}
		\centering
		\includegraphics[width=0.8\linewidth]{figures/circPathVel.eps}
		\caption{Speed Profile along Circular Path}
		\label{fig:circPathSpeed}
		\end{figure}
		\end{comment}
		%	\begin{lemma}
		%		The optimal time $\tau_c$ given in \eqref{eq:circ_path_opt_time} is a convex function of the terminal speeds $v_{0}$ and $v_{f}$.
		%	\end{lemma}
		%	\begin{proof}
		%		See Appendix \ref{ap:proof_convex_opt_time}
		%		The Hessian of $\tau_c$ is:
		%		\be
		%		H(v_{0},v_{f})=\bbmat h_{11} & h_{12} \\
		%		h_{21} & h_{22}\ebmat
		%		\ee
		%		
		%	\end{proof}
		
		For each segment on the path, one can find the $\calC^0$ velocity profile that satisfies acceleration bound and approximately minimizes the total travel time for the given terminal speeds. 
		%Furthermore, the corresponding times are convex functions of these terminal speeds. 
		Let the total time required to travel $\mathbf{P}_0^f$ under $\tilde{\mathbf{v}}$ be:
		\be
		\arraycolsep=1pt
		\baa
		\tau \left(\mathbf{P}_0^f,\tilde{\mathbf{v}} \right)=&  \ds \sum_{l=1}^{N_{\mathbf{P}}}    \tau_{i_l}(\Gamma_l,v_{l},v_{l+1})  
		\eaa
		\ee
		where $i_l=\frac{3+(-1)^l}{2}$, $N_{\mathbf{P}}=2 N_{\mathbf{P}}^c+2$. % at the intersections of the straight line and circular arc segments along the path as shown in Fig.~\ref{fig:shortest_path}.
		The terminal speed vector $\tilde{\mathbf{v}}$ is found by solving:
		\be \label{eq:opt_time_prob}
		\baa
		\text{Minimize} & \tau \left(\mathbf{P}_0^f,\tilde{\mathbf{v}} \right) \\
		\text{Subject to} 	
		& 1)\; \Gamma_{l}\ge 
		\begin{cases}
			\log\left(\frac{\bar{u}-C_D v_{l}^2}{\bar{u}-C_D v_{l+1}^2} \right), & \text{if } v_{l}<v_{l+1}\\
			\log\left(\frac{\bar{u}+C_D v_{l}^2}{\bar{u}+C_D v_{l+1}^2}\right), & \text{if } v_{l}\ge v_{l+1}
		\end{cases}\\
		& \quad \; \forall l \in \{1,3,...,N_{\mathbf{P}}\}\\
		& 2)\; \Gamma_{l}\ge 
		\begin{cases}
			\Gamma^{+}(v_{l},v_{l+1}), & \text{if } v_{l}<v_{l+1}\\
			\Gamma^{-}(v_{l},v_{l+1}), & \text{if } v_{l}\ge v_{l+1}\\
		\end{cases} \\
		& \quad \; \forall l \in \{2,4,...,N_{\mathbf{P}}\}\\
		& 3)\;  v_{l} \in [0, \bar{v}^c], \forall l \in \{2,3,...,N_{\mathbf{P}}-1\}
		\eaa
		\ee
		where $\Gamma^{+}(v_{1},v_{2})=\frac{\rho_{\bar{o}} \left (\tanh^{-1}\left(\eta^+(v_2)\right)-\tanh^{-1}\left(\eta^+(v_1)\right)\right)}{\lambda_0}$,  $\Gamma^{-}(v_{1},v_{2})=\frac{\rho_{\bar{o}} \left (\tanh^{-1}\left(\eta^-(v_2)\right)-\tanh^{-1}\left(\eta^-(v_1)\right)\right)}{\lambda_0}$, $\eta^+(v)= \frac{C_D\rho_{\bar{o}}^2 \bar{u}+2(v)^2}{\lambda_0(\bar{v}^c)^2}$ and $\eta^-(v)= \frac{C_D\rho_{\bar{o}}^2 \bar{u}-2(v)^2}{\lambda_0(\bar{v}^c)^2}$, and maximum possible speed on the circular segment $\bar{v}^c = \sqrt{\frac{\sqrt{C_D^2\rho_{\bar{o}}^4+4\bar{u}^2\rho_{\bar{o}}^2}-C_D\rho_{\bar{o}}^2}{2}}$. The near time-optimal velocity profile is determined as:
		\be
		\tilde{\mathbf{v}}_{0}^{f} = \argmin_{\tilde{\mathbf{v}}}\left(\tau \left(\mathbf{P}_0^f,\tilde{\mathbf{v}} \right)\right).
		\ee
		%The optimization problem in \eqref{eq:opt_time_prob} is highly non-linear. However, interestingly, the problem has a very simple solution. 
		The solution to \eqref{eq:opt_time_prob} will have maximum possible speeds in the feasible set along the given path to ensure that the total travel time is minimized. For simplicity, we chose $v_1=v_{2N_{\mathbf{P}}^c+2}=0$, however this approach can be applied to any feasible non-zero initial and final speeds. The terminal speed vector that maximizes the speeds along the path with the given acceleration constraints can be obtained using Algorithm \ref{alg:min_time_vel}.
		\begin{algorithm}[h]
			\caption{Near time-optimal velocity on path $\mathbf{P}$}
			\label{alg:min_time_vel}
			%\DontPrintSemicolon % Some LaTeX compilers require you to use \dontprintsemicolon    instead
			\KwIn{$\mathbf{P}$, $N_{\mathbf{P}}^c$, $\bar{u}$, $C_D$}
			\KwOut{$\bar{\mathbf{v}}=[0,v_2,v_3,...,v_{2N_{\mathbf{P}}^c+1},0]$}
			$i=1$, $j=2N_{\mathbf{P}}^c+2$; $v_i=v_j=0$;\\
				\While{$v_i<\bar{v}^c \; \&  \;i <=2N_{\mathbf{P}}^c+2$}{
				    $v_{i+1}=\sqrt{\frac{1}{2}\left (\kappa_1\tanh(\frac{\Gamma_i \lambda_0}{\rho_{\bar{o}}}+\tanh^{-1}(\frac{(\kappa_2+2v_i^2)}{\kappa_1}))-\kappa_2 \right) }$\\
					\If{$i$ is odd (i.e., straight segment)}
					{ 
					$v_{i+1}=\sqrt{\frac{(\bar{u}-e^{-2 C_D \Gamma_i} (\bar{u}-C_D v_i^2))}{C_D}}$\\
						\If{acceleration beyond $\bar{v}^c$ possible}
						{$v_{i+1}=\bar{v}^c$}
					}				
					$i=i+1$;		
				}
			
				\While{$v_j<\bar{v}^c$}{
				    $v_{j-1}=\sqrt{\frac{1}{2}\left (\kappa_2-\kappa_1\tanh(-\frac{\Gamma_i \lambda_0}{\rho_{\bar{o}}}+\tanh^{-1}(\frac{(\kappa_2-2v_i^2)}{\kappa_1})) \right) }$\\
					\If{$j-1$ is odd}
					{   $v_{j-1}=\sqrt{\frac{(-\bar{u}+e^{-2 C_D \Gamma_i} (\bar{u}+C_D v_i^2))}{C_D}}$\\
						\If{deceleration to $0$ from $\bar{v}^c$ possible}
						{$v_{j-1}=\bar{v}^c$}
					}
					$j=j-1$;		
				}
				$v_{k}=\bar{v}^c$ for all k such that $i<k<j$\\
			\Return{$\bar{\mathbf{v}}=[v_1,v_2,v_3,...,v_{2N_{\mathbf{P}}^c+1},v_{2N_{\mathbf{P}}^c+2}]'$}
		\end{algorithm}
		
	%Appendixes should appear before the acknowledgment.
	
	\section{Conclusion}\label{sec:conclusions}
	We developed $\calC^1$-Tangent graph for an obstacle environment which provides continuously differentiable ( $\calC^1$) path between any two nodes. We developed a novel quadratic function to find common tangents of the obstacle boundaries in $\calC^1$-Tangent graph. We also extended the ellipse filter to more generic environments with convex obstacles to reduce search time of an algorithm finding shortest path on $\calC^1$-Tangent graph.
	
	We provide analytical expressions for the near time-optimal velocity profiles for the agents moving on a shortest path obtained on $\calC^1$-Tangent graph under damped double integrator dynamics with bounded acceleration. 
	
	\vspace{-1mm}
	\bibliographystyle{IEEEtran}
	\bibliography{supplement_Refs}
\end{document}